\documentclass[conference]{IEEEtran}
\IEEEoverridecommandlockouts
\usepackage{cite}
\usepackage{amsmath,amssymb,amsfonts}
\usepackage{graphicx}
\usepackage{textcomp}
\usepackage{tikz}
\usepackage{xcolor}
\usepackage{multirow}
\usetikzlibrary{arrows.meta,calc}

\usepackage{amsthm}

\usepackage{booktabs}
\usepackage{caption}
\usepackage{cite}
\usepackage{colortbl}
\usepackage{pdfpages}
\usepackage{dsfont}

\usepackage{bm}
\usepackage{bbm}
\usepackage{dirtytalk}

\usepackage{subcaption}

\newtheorem{theorem}{Theorem}

\newtheorem{lemma}{Lemma}

\theoremstyle{definition}

\newtheorem{definition}{Definition}

\theoremstyle{definition}
\newtheorem{remark}{Remark}

\theoremstyle{definition}

\definecolor{DarkGreen}{rgb}{0.1,0.5,0.1}
\definecolor{DarkRed}{rgb}{0.5,0.1,0.1}
\definecolor{DarkBlue}{rgb}{0.1,0.1,0.5}
\definecolor{DarkPurple}{rgb}{0.5,0.2,0.5}
\definecolor{DarkTurquoise}{rgb}{0.1,0.5,0.5}

\newcolumntype{P}[1]{>{\centering\arraybackslash}p{#1}}

\newcommand\blfootnote[1]{%
  \begingroup
  \renewcommand\thefootnote{}\footnote{#1}%
  \addtocounter{footnote}{-1}%
  \endgroup
}

\DeclareMathOperator{\tr}{tr}

\usepackage{booktabs}
\usepackage{diagbox}

\def\BibTeX{{\rm B\kern-.05em{\sc i\kern-.025em b}\kern-.08em
    T\kern-.1667em\lower.7ex\hbox{E}\kern-.125emX}}
\begin{document}

\title{\textit{HerA} Scheme: Secure Distributed Matrix Multiplication via Hermitian Codes\\
}

\author{\IEEEauthorblockN{Roberto A. Machado} \IEEEauthorblockA{School of Mathematical and Statistical Sciences\\ Clemson University\\ Clemson, SC 29634 USA\\rmachad@clemson.edu}
\and
\IEEEauthorblockN{Gretchen L. Matthews 
}\IEEEauthorblockA{Department of Mathematics\\
Virginia Tech\\
Blacksburg, VA 24061 USA\\
gmatthews@vt.edu}
\and
\IEEEauthorblockN{Welington Santos} \IEEEauthorblockA{Department of Mathematics\\
Virginia Tech\\
Blacksburg, VA 24061 USA\\
welington@vt.edu}

}

\maketitle

\blfootnote{The work of Gretchen L. Matthews was partly supported by NSF DMS-2037833, NSF DMS-2201075, and the Commonwealth Cyber Initiative.}
\begin{abstract}

We consider the problem of secure distributed matrix multiplication (SDMM), where a user has two matrices and wishes to compute their product with the help of $N$ honest but curious servers under the security constraint that any information about either $A$ or $B$ is not leaked to any server. This paper presents a \emph{new scheme} that considers the inner product partition for matrices $A$ and $B$.
Our central technique relies on encoding matrices $A$ and $B$ in a Hermitian code and its dual code, respectively. We present the Hermitian Algebraic (HerA) scheme, which employs Hermitian codes and characterizes the partitioning and security capacities given entries of matrices belonging to a finite field with $q^2$ elements. We showcase that this scheme performs the secure distributed matrix multiplication in a significantly smaller finite field and expands security allowances compared to the existing results in the literature. 
\end{abstract}
\vspace{-10pt}

\section{Introduction}

Matrix multiplication is an essential back-end operation of numerous applications in signal processing and machine learning. When facing applications involving massive matrices, matrix multiplication in a single computer is slow, and distributed solutions need to be adopted. In such a scenario, the goal is to speed up the computational time to perform the matrix multiplication. Thus, the multiplication task is divided into smaller sub-tasks distributed across dedicated workers. 

The setting for the problem considered in this paper is as follows. A user has two matrices, $A \in \mathbb{F}^{a \times b}_{q^{2}}$ and $B \in \mathbb{F}^{b \times c}_{q^{2}}$, and wishes to compute their product, $AB \in \mathbb{F}_{q^2}^{a \times c}$, with the assistance of $N$ servers, without leaking any information about either $A$ or $B$ to any server.  We assume that all servers are honest but curious (passive) in that they are not malicious and will follow the pre-agreed-upon protocol. However, any $T$ may collude to eavesdrop and deduce information about either $A$ or $B$.

We follow the  setting proposed in~\cite{ravi2018mmult}, with many follow-up works~\cite{Kakar2019OnTC,koreans,d2019gasp,DOliveira2019DegreeTF, Aliasgari2019DistributedAP,aliasgari2020private,Kakar2019UplinkDownlinkTI,doliveira2020notes,Yu2020EntangledPC,mital2020secure,bitar2021adaptive, hasircioglu2021speeding, danielgretchen, machado, 9965858, 9004505, bivariatepol}. The performance metric initially used was the download cost, i.e., the total amount of data downloaded by the users from the server. Subsequent work has also considered the upload cost \cite{mital2020secure}, the total communication cost \cite{9004505, machado}, and computational costs \cite{doliveira2020notes, 9965858}.

Different partitionings of the matrices lead to different trade-offs between upload and download costs. In this paper, we consider the inner product partitioning given by $A = \begin{bmatrix}A_1 & \cdots & A_L\end{bmatrix}$ and $B^{\intercal} = \begin{bmatrix}B_1^{\intercal} & \cdots & B_L^{\intercal}\end{bmatrix}$ such that $AB = A_1 B_1 + \cdots + A_L B_L$, where all products $A_\ell B_\ell$ are well-defined and of the same size. Under this partitioning, a polynomial code is a polynomial $h(x,y)=f(x,y) \cdot g(x,y)$, whose coefficients encode the sub-matrices $A_kB_\ell$. The $N$ servers compute the evaluations $h(P_1),\ldots,h(P_N)$ for certain points $P_1, \ldots, P_N$ in an Hermitian curve. The servers send these evaluations to the user. The two-variable polynomial $h(x,y)$ is constructed to ensure that no $T$-subset of evaluations reveals any information about $A$ or $B$ ($T$-security), and the user can reconstruct $AB$ given all $N$ evaluations $h(P_1)$, $\ldots$, $h(P_N)$(decodability).

Examples of polynomial schemes using the inner product partitioning are the secure MatDot codes in \cite{Aliasgari2019DistributedAP}, the DFT-codes in~\cite{mital2020secure}, and the FTP codes \cite{machado}. Some authors started exploring two-variable polynomials in the context of secure distributed matrix multiplication using outer product partitioning, \cite{bivariatepol, hasircioglu2021speeding}. One of the literature's main focuses was minimizing the minimum amount of helping servers $N$, also known as the recovery threshold, to reduce the communication cost. In \cite{machado}, Machado \textit{et al.} presented a scheme to reduce the total communication by contacting more servers. Most of the constructions rely on large finite fields and even extensions of finite fields. This paper investigates the partitioning and security capacities given matrices $A$ and $B$ have entries in $\mathbb{F}_{q^2}$, a finite field with $q^2$ elements.

We present the Hermitian Algebraic (\textit{HerA}) scheme, a two-variable polynomial scheme inspired by Algebraic Codes in secret sharing schemes literature, specifically the Algebraic Codes for Secret Sharing Schemes were first introduced in \cite{umbertosecsharing}, a protocol close to optimal communication efficiency and robust security with lengths not bounded by the field size. When employing this scheme to the secure matrix multiplication problem, matrix $A$ should be encoded in a Hermitian code while matrix $B$ is encoded in its dual. Therefore, the recovery threshold is allowed to be larger than the field's size ${q^2}$,  which no other polynomial scheme could achieve.

\begin{figure}[!t]
\centering
\begin{tikzpicture}
\tikzset{
myline/.style={-{Latex[length=2mm, width=1.3mm]}},
mynode/.style={pos=.5,fill=white,draw,text width=24mm,align=center,inner sep=1pt}
}

\node[draw] (U1) at (0,0) {\small User};
\node[draw] (S1) at (4,2.1) {\small Server 1};
\node[draw] (S2) at (4,0.7) {\small Server 2};
\node[draw] (S3) at (4,-0.7) {\small Server 3};
\node[draw] (S4) at (4,-2.1) {\small Server 4};
\node[draw] (U2) at (8,0) {\small User};
\coordinate (BL1) at ($(U1 |- S4)+(-0.5,0)$);
\coordinate (TR1) at ($(S1)+(-0.8,0)$);
\coordinate (BL2) at ($(S4)+(0.8,0)$);
\coordinate (TR2) at ($(U2 |- S1)+(0.5,0)$);

\draw[myline] (U1.east) -- (S1.west) node[mynode] {\footnotesize$f(P_{1\delta}), g(P_{1\delta})$};
\draw[myline] (U1.east) -- (S4.west) node[mynode] {\footnotesize$f(P_{\delta\delta^2}), g(P_{\delta\delta^2})$};
\draw[myline] (U1.east) -- (S2.west) node[mynode] {\footnotesize$f(P_{\delta^2\delta}), g(P_{\delta^2\delta})$};
\draw[myline] (U1.east) -- (S3.west) node[mynode] { \footnotesize$f(P_{\delta^2\delta^2}), g(P_{\delta^2\delta^2})$};

\draw[myline] (S1.east) -- (U2.west) node[mynode] {\footnotesize$h(P_{1\delta})$};
\draw[myline] (S2.east) -- (U2.west) node[mynode] {\footnotesize$h(P_{\delta\delta^2})$};
\draw[myline] (S4.east) -- (U2.west) node[mynode] {\footnotesize$h(P_{\delta^2\delta})$};
\draw[myline] (S3.east) -- (U2.west) node[mynode] {\footnotesize$h(P_{\delta^2\delta^2})$};

\draw[dotted] (BL1) rectangle (TR1);
\draw[dotted] (BL2) rectangle (TR2);
\node[above] at ($(BL1)!0.5!(BL1 -| TR1)$) {Upload Phase};
\node[above] at ($(BL2)!0.5!(BL2 -| TR2)$) {Download Phase};
\end{tikzpicture}
\caption{An example of the \textit{HerA} Scheme detailed in Section \ref{sec3}. The user computes carefully chosen evaluations of the two-variable polynomials $f(x,y)$ and $g(x,y)$ and uploads them to the servers. This allows for  reducing the minimum required size of a finite field. Then, each server computes the product of their received evaluations, which is itself an evaluation of the polynomial $h(x) = f(x) \cdot g(x)$ and sends it back to the user who can decode $AB$.}\label{fig:phase1and2}
\end{figure}
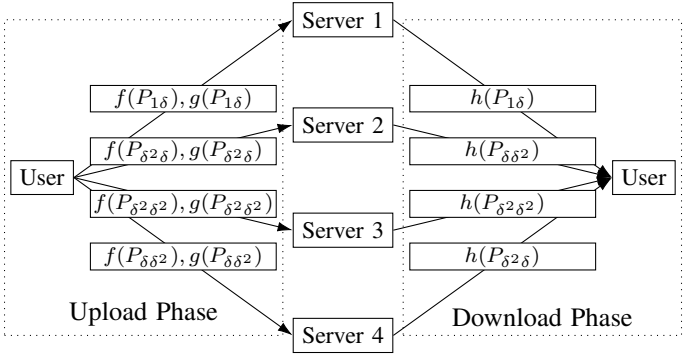


\begin{theorem}\label{theo:scheme}
Let $L$ and $T$ be positive integers. Let $A \in \mathbb{F}_{q^2}^{a \times b}$, $B \in \mathbb{F}_{q^2}^{b \times c}$ be two matrices and suppose the $T$-MDS condition is fulfilled. 
Then, there exists a \textit{HerA} scheme with partitioning parameter $L$ and security parameter $T$, which securely computes $AB\in\mathbb{F}_{q^2}^{a\times c}$ utilizing $L+2T$ servers with a total communication rate of
\begin{align} \label{eq: FTP rate}
\mathcal{R} = \left(\frac{Nb}{L} \left( \frac{1}{a}+\frac{1}{c} \right)+N\right)^{-1} .
\end{align}
\end{theorem}



\subsection{Related Work}

For distributed computations, polynomial codes were initially introduced in~\cite{polycodes1} to mitigate stragglers in distributed matrix multiplication. A series of works followed this,~\cite{polycodes2,pulkit,pulkit2,fundamental}.

The literature on SDMM has also studied different variations of the model we focus on here. For instance, in~\cite{nodehi2018limited,jia2019capacity,mital2020secure,akbari2021secure}, the encoder and decoder are considered to be separate parties, in~\cite{nodehi2018limited} servers are allowed to cooperate, and in~\cite{kim2019private} the authors consider a hybrid between SDMM and private information retrieval where the user has a matrix $A$ and wants to privately multiply it with a matrix $B$ belonging to some public list. \textit{HerA} codes can be readily used or adapted to many of these settings as done with other polynomial schemes (e.g.,~\cite{bitar2021adaptive, zhu2021improved}).

Algebraic curves, including Hermitian ones, have gained some attention in secret sharing and fractional decoding (e.g., \cite{senthoor2022concatenating, umbertosecsharing, cheraghchi2019nearly, matthews2021fractional, murphy2022codes }). The field trace method relevant to us was developed in \cite{umbertosecsharing} and later extended to Communication Efficient Quantum Secret Sharing \cite{senthoor2022concatenating}.

\subsection{Main Contributions}

Our main contributions are summarized below.

\begin{itemize}

    \item We present a \textit{new} polynomial encoding scheme (called \textit{HerA} scheme) for the secure distributed matrix multiplication problem, considering the inner product partition and rational functions in the Riemann-Roch space associated with a divisor of the Hermitian curve. \textit{HerA} scheme allows for performing $AB$ in a smaller finite field compared to the state of art in the literature.
    
    \item By carefully selecting a divisor $mP_\infty$, the matrix $A$ is encoded in a Hermitian code while matrix $B$ is encoded in its dual code, which is, by construction, also a Hermitian code. This allows us to use the inner product property of dual codes and achieve a recovery threshold of $L+2T$, the same as state of the art in the literature, Theorem~\ref{theo:scheme}.

    \item As we illustrated in the examples, in Sections \ref{sec3} and \ref{sec6}, \textit{HerA} can perform information-theoretic secure distributed matrix multiplication over finite fields smaller or equal to the recovery threshold.

\end{itemize}

\section{Preliminaries}

This section introduces some basic notation and the main results in Hermitian codes needed for the rest of the paper. For example, we define $[M,N] = \{M,M+1,\ldots,N\}$ and $[M] = [1,M]$.

We record some facts about Hermitian codes from \cite{Stichtenoth2009AlgebraicFF}.

\noindent For a prime power $q$, let $\mathcal{H}_{q}$ denote the smooth, projective curve given by $y^{q}+y=x^{q+1}$ over the finite field $\mathbb{F}_{q^2}$. The genus of $\mathcal{H}_{q}$ is $g=\frac{q(q-1)}{2}$, and there are $q+1$ distinct $\mathbb{F}_{q^2}$-rational places.

Let $P_{\infty},P_{1},\ldots,P_{n}$ be the $n+1$ distinct $\mathbb{F}_{q^2}$-rational places so that $n=q^3$. Given $\alpha\in\mathbb{F}_{q^2}$, consider $\Gamma_{\alpha}:=\left\lbrace\beta\in\mathbb{F}_{q^2}:\beta^{q}+\beta=\alpha^{q+1}\right\rbrace$. It is well known that for all $\alpha\in\mathbb{F}_{q^2}$, $\mid\Gamma_{\alpha}\mid=q$ and that the affine rational points of $\mathcal{H}_{q}$ are of the form $P_{\alpha\beta}:=(\alpha,\beta)\in\mathbb{F}_{q^2}\times\Gamma_{\alpha}$; that is, the set of $\mathbb{F}_{q^2}$ rational points of $\mathcal{H}_{q}$ is
\begin{equation*}
\mathcal{H}_{q}(\mathbb{F}_{q^2}):=\left\lbrace P_{\alpha\beta}:\alpha\in\mathbb{F}_{q^2}, \beta\in\Gamma_{\alpha}\right\rbrace\cup\left\lbrace P_{\infty}\right\rbrace,
\end{equation*}
where $P_{\infty}$ denotes the unique point at infinity which has projective coordinates $(0:1:0)$.
Recall that the Riemann-Roch space of a divisor $mP_{\infty}$ on $\mathcal{H}_{q}$ is the subset $\mathcal{L}(mP_{\infty})$ of $\mathbb{F}_{q^2}[x,y]$ generated by $I(m)$, where
\begin{equation*}
 I(m)=\lbrace x^{i}y^{j}: 0\leq i, 0\leq j\leq q-1, iq+j(q+1)\leq m\rbrace.
\end{equation*}
The one-point Hermitian code with design parameter $m$ is the algebraic geometry code $\mathcal{C}(mP_{\infty}):=ev(\mathcal{L}(mP_{\infty}))$; that is, 
\begin{equation}
\mathcal{C}(mP_{\infty}):=\left\lbrace(f(P_{1}),\ldots,f(P_{n})):f\in\mathcal{L}(mP_{\infty})\right\rbrace.
\end{equation}
\noindent Note that $\mathcal{C}(mP_{\infty})$ is a linear code of length $n=q^3$ over the field $\mathbb{F}_{q^2}$ and for $m^\prime <m$ we have $\mathcal{C}(m^{\prime}P_{\infty})\subseteq\mathcal{C}(mP_{\infty})$. Moreover, $\mathcal{C}(mP_{\infty})=\lbrace 0\rbrace$ for $m<0$, and $\mathcal{C}(mP_{\infty})=\mathbb{F}^{n}_{q^2}$ for $m>q^3+q^2-q-2$.

\begin{remark}\label{DualofHermitian} \cite[Proposition 8.3.2]{Stichtenoth2009AlgebraicFF}

\noindent For $0\leq m\leq q^3+q^2-q-2$, the dual code of the Hermitian code $\mathcal{C}(mP_{\infty})$ is
\begin{equation}
\mathcal{C}(mP_{\infty})^{\perp}=\mathcal{C}(m^{\perp}P_{\infty}),
\end{equation}
where $m^{\perp}=q^3+q^2-q-2-m$.
\end{remark}
\noindent Remark \ref{DualofHermitian} implies that $\mathcal{C}(mP_{\infty})$ is self-orthogonal if $2m\leq q^3+q^2-q-2$, and $\mathcal{C}(mP_{\infty})$ is self-dual if $m=\frac{q^3+q^2-q-2}{2}$.

\begin{lemma}\cite[Proposition 8.3.3]{Stichtenoth2009AlgebraicFF}

\noindent Suppose that $0\leq m\leq s:=q^3+q^2-q-2$. Then the following hold:
\begin{itemize}
    \item[i)]$\dim\mathcal{C}(mP_{\infty})=\left\lbrace\begin{array}{ll}
         \mid I(m)\mid&\text{for } 0\leq m\leq q^3  \\
         q^3-\mid I(m^{\perp})\mid&\text{for } q^3\leq m\leq s. 
    \end{array}\right.$

\item[ii)]For $q^2-q-2<m<q^3$ we have
\begin{equation*}
\dim\mathcal{C}(mP_{\infty})=m-\frac{q(q-1)}{2}+1.
\end{equation*}
\item[iii)] The minimum distance of $\mathcal{C}(mP_{\infty})$ is $d\geq d^{\star}=q^3-m$. 
\end{itemize}
\end{lemma}
\noindent The value $d^{\star}$ is called the \textit{designed minimum distance} of $\mathcal{C}(mP_{\infty})$. If $q^{2}-q-1<m<q^3-q^2+q+1$ then $d=d^{\star}=q^3-m$.

Theoretically, a Hermitian code can also be constructed by evaluating $f(x,y)\in\mathcal{L}(mP_{\infty})$ at a proper subset of the affine rational points, but then Remark \ref{DualofHermitian} may no longer hold. 

\section{A Motivating Example: $L=2$ and $T=1$} 

\label{sec3}

We begin our description of \textit{HerA} (Hermitian Algebraic) codes with the following example, which we present in as much detail as possible to present the crucial components of the scheme. We compare the size of the field required for a Discrete Fourier Transform code and a \textit{HerA} code to illustrate the flexibility \textit{HerA} codes have.

In this example, a user desires to compute the product of two matrices $A = \begin{bmatrix}A_1 &A_2\end{bmatrix}\in \mathbb{F}_{5}^{a \times b}$ and $B$ where $B^{\intercal} = \begin{bmatrix}B_1^{\intercal} & B_2^{\intercal}\end{bmatrix}\in \mathbb{F}_{5}^{c \times b}$ such that $AB = A_1 B_1 + A_2 B_2$ with the assistance of non-colluding helper servers. The solution via Discrete Fourier Transform codes schemes utilizes $N=4$ servers. It involves picking two random matrices $R \in \mathbb{F}_{5}^{a \times \frac{b}{2}}$ and $S \in \mathbb{F}_{5}^{\frac{b}{2} \times c}$ and constructing the one-variable polynomials $f'(x) = A_1 + A_2x + Rx^2$ and $g'(x) = B_1 + B_2x^{-1} + Sx^{-3}$. The user then selects four distinct non-zero elements $\beta_1,\beta_2,\beta_3, \beta_4 \in \mathbb{F}_{5}$ and uploads both $f'(\beta_i)$ and $g'(\beta_i)$ to Server $i$. Each server then computes the product $f'(\beta_i) \cdot g'(\beta_i)$. This is equivalent to computing an evaluation $h'(\beta_i)$ of the polynomial $h'(x) = (A_1B_1 + A_2B_2) + (A_2B_1 + RB_2)x + RB_1x^2 + A_1S x^{-3} + A_2S x^{-2} + (A_1B_2+RS) x^{-1}$. The user then downloads each $h'(\beta_i)$, obtaining four evaluations of a polynomial of degree two. Therefore, the user can retrieve $AB = A_1B_1 + A_2B_2$ by operating $4\sum_{i=1}^4 h'(\beta_i)$ since $\sum_{i=1}^4 (3^{i})^s = 0 \ \  \forall s : 4 \nmid  s$.

The security of the Discrete Fourier Transform codes follows from the fact that $I(f'(\beta_i), g'(\beta_i);A,B)=0$. As for the communication costs, first the user uploads $f'(\beta_i)$ and $g'(\beta_i)$, which cost $ab$ and $bc$, symbols respectively, four times. Thus, the upload cost is $4(ab+bc)$ symbols of $\mathbb{F}_{5}$. Then, the user downloads $h'(\beta_i)$, which costs $ac$ symbols of $\mathbb{F}_{5}$, four times, obtaining a download cost of $4ac$ symbols of $\mathbb{F}_{5}$. Since the user retrieves $AB \in \mathbb{F}_{5}^{a \times c}$, which consists of $ac$ symbols of $\mathbb{F}_{5}$, the total communication rate is given by $\mathcal{R}' = \frac{ac}{4ab+4bc+4ac}$.

The setting we consider for our construction is similar to the one considered for DFT codes except for the size of the field, i.e., a user wants to compute the product of two matrices $A = \begin{bmatrix}A_1 &A_2\end{bmatrix}\in \mathbb{F}_{4}^{a \times b}$ and $B$ where $B^{\intercal} = \begin{bmatrix}B_1^{\intercal} & B_2^{\intercal}\end{bmatrix}\in \mathbb{F}_{4}^{c \times b}$ such that $AB = A_1 B_1 + A_2 B_2$. We also pick two random matrices $R \in \mathbb{F}_{4}^{a \times \frac{b}{2}}$ and $S \in \mathbb{F}_{4}^{\frac{b}{2} \times c}$. Based on properties of Hermitian codes, we present an \textit{HerA} code which allows multiplying matrices in a smaller finite field utilizing $N=4$ servers.

Let $\delta\in\mathbb{F}_4$ denote a algebraic element in $\mathbb{F}_4$ such that $\delta^2 + \delta +1 = 0$. Let $P_{\alpha\beta} \in \mathbb{F}_{4}\times \Gamma_\alpha$ denote rational points satisfying $\beta^2 + \beta = \alpha^3$.  Therefore, $P_{00} = (0, 0)$, $P_{01} = (0, 1)$, $P_{1\delta} = (1, \delta)$, $P_{1\delta^2} = (1, \delta^2)$, $P_{\delta\delta} = (\delta, \delta)$, $P_{\delta\delta^2} = (\delta, \delta^2)$, $P_{\delta^2\delta} = (\delta^2, \delta)$ and  $P_{\delta^2\delta^2} = (\delta^2, \delta^2)$ are all the possible affine rational points. Let $f(x, y)$ be a two-variable polynomial generated by monomials $\{1, x, y\}$ such that $f(P_{00}) = A_1$, $f(P_{01}) = A_2$  and $f(P_{1\delta}) = R$. Let $g(x,y)$ be a two-variable polynomial generated by monomials $\{1, x, x^2, y, xy\}$ such that $g(P_{00}) = B_1$, $g(P_{01}) = B_2$, $g(P_{1\delta}) = S$, $g(P_{1\delta^2}) = 0$ and $g(P_{\delta\delta}) = 0$. The explicit polynomials are

$$f(x, y) = A_1 + (A_1+R+(A_1+A_2)\delta)x+(A_1+A_2)y$$ and 
$$\begin{array}{rc}g(x,y)= B_1 + (S+B_1\delta^2+B_2)x+(B_1\delta+B_2+S\delta)x^2&\\
+(B_1+B_2)y+(B_1+B_2+S)xy&
\end{array}$$

Then, $h(x,y) = f(x,y) \cdot g(x,y)$ is such that $h(P_{00}) + h(P_{01})= A_1B_1 + A_2B_2 = AB$.

Since $(f(P_{00}), f(P_{01}), \ldots, f(P_{\delta^2\delta^2})) \in \mathcal{C}(3P_\infty)$ and $$\left(g(P_{00}), g(P_{01}), \ldots, g(P_{\delta^2\delta^2})\right) \in \mathcal{C}(5P_\infty),$$
the dual of $\mathcal{C}(3P_\infty)$, it follows that 

\begin{equation}\label{eq:ex_dual}
    \displaystyle \sum_{\alpha\beta} f(P_{\alpha\beta}) \cdot g(P_{\alpha\beta})  = \sum_{\alpha\beta} h(P_{\alpha\beta})= 0.
\end{equation}

Our scheme works as follows: the user uploads the evaluations $f(P_i)$ and $g(P_i)$ to each Server $i$, where $(P_1, P_2, P_3, P_4) = (P_{1\delta},P_{\delta\delta^2}, P_{\delta^2\delta},P_{\delta^2\delta^2})$. Then, each Server $i$ computes $-h(P_i)$, and sends it back to the user. The user can decode $AB$ as follows:

\begin{align*}
    - & h(P_{1\delta}) - h(P_{\delta\delta^2}) - h(P_{\delta^2\delta}) - h(P_{\delta^2\delta^2}) \\
    =& \!  -\! h(P_{1\delta}) \!-\! h(P_{1\delta^2}) \!-\! h(P_{\delta\delta}) \!- \!h(P_{\delta\delta^2})\!-\! h(P_{\delta^2\delta}) \! - \! h(P_{\delta^2\delta^2})\\
    =&  h(P_{00}) + h(P_{01}) = A_1B_1 + A_2B_2 =  AB.
\end{align*} 

The first equality follows from the fact that $h(P_{1\delta^2}) =f (P_{1\delta^2})\cdot  g(P_{1\delta^2}) = 0 = f(P_{\delta\delta})\cdot g(P_{\delta\delta}) = h(P_{\delta\delta})$ since $g(P_{1\delta^2}) = g(P_{\delta\delta}) = 0$. The second equality follows from Equation \ref{eq:ex_dual}.

Security follows by showing that $I(f(P_i), g(P_i);A,B)=0$, as is done in Lemma \ref{lem:tsecure}. As for the communication costs, first, the user uploads $f(P_i)$ and $g(P_i)$, which cost $2ab$ and $2bc$, symbols, respectively. Thus, the upload cost is $(2ab+2bc)$ symbols of $\mathbb{F}_{4}$. Then, the user downloads $h(P_i)$, which costs $ac$ symbols of $\mathbb{F}_{4}$, four times, obtaining a download cost of $4ac$ symbols of $\mathbb{F}_{4}$. Since the user retrieves $AB \in \mathbb{F}_{4}^{a \times c}$, which consists of $2ac$ symbols of $\mathbb{F}_{4}$, the total communication rate is given by $\mathcal{R} = \frac{ac}{4ab+4bc+2ac}$. 

We note that the total communication rate of the \textit{HerA} code is equal to the total communication rate of the DFT code. However, we showcase that the \textit{HerA} code utilizes a smaller field size avoiding the divisibility constraint required for DFT codes. It raises the theoretical question on the field's capacity: Is $\mathbb{F}_4$ the 
smallest field to perform the product of matrices $A$ and $B$ given $L=2$ and $T=1$? On the other side, what are the maximum partitioning and security parameters allowed in secure distributed matrix multiplication over $\mathbb{F}_4$?

\section{HerA Scheme}

This section presents the general construction for the \textit{HerA} Scheme. The main idea is to perform the same technique as in Section~\ref{sec3} that retrieves the product $AB$ using the inner product and encoding  matrix $A$ in a Hermitian code while matrix $B$ is encoded in its dual code. Consider the matrices $A\in\mathbb{F}_{q^2}^{a\times b}$, $B\in\mathbb{F}_{q^2}^{b\times c}$.

\noindent \textbf{Choosing Parameters $L$, $T$ and $m$:} We begin by choosing parameters $L$ and $T$ such that $2(L+T)\leq q^3 - \frac{q(q-1)}{2}$ and set $m=L+T+\frac{q(q-1)}{2}-1$. We remark that the bound $2(L+T)\leq q^3 - \frac{q(q-1)}{2}$ explicates a finite field's partitioning and security capacities with $q^2$ elements.

\noindent \textbf{Choosing the Polynomials:} As described in the introduction, we consider the setting where the user partitions the matrices $A\in\mathbb{F}^{a\times b}_{q^2}$ and $B\in\mathbb{F}^{b\times c}_{q^2}$ as $A=[A_{1}\ldots A_{L}]$ and as $B^{T}=[B^{T}_{1}\ldots B^{T}_{L}]$ such that $AB=A_{1}B_{1}+\cdots+A_{L}B_{L}$, where each $A_{i}\in\mathbb{F}^{a\times\frac{b}{L}}_{q^2}$, $B_{i}\in\mathbb{F}^{\frac{b}{L}\times c}_{q^2}$. In order to obtain $T$-security, $R_{1},\ldots,R_{T}\in\mathbb{F}^{a\times\frac{b}{L}}_{q^2}$ and $S_{1},\ldots,S_{T}\in\mathbb{F}^{\frac{b}{L}\times c}_{q^2}$ are chosen independently and uniformly at random. We then choose $\left\lbrace P_{1},\ldots,P_{L+T}\right\rbrace\subseteq\mathcal{H}_{q}(\mathbb{F}_{q^2})\setminus\left\lbrace P_{\infty}\right\rbrace$, $f\in\mathcal{L}(mP_{\infty})$, and $g\in\mathcal{L}(m^{\perp}P_{\infty})$ such that $f(P_{i})=A_{i}$, $g(P_{i})=B_{i}$ for every $i\in [L]$; $f(P_{L+i})=R_{i}$, $g(P_{L+i})=S_{i}$ for every $i\in [T]$; and $g(P_{L+T+i})=0$ for every $i\in [q^{3}-2(L+T)]$.

\noindent \textbf{Upload Phase:} The \textit{HerA} scheme uses $L+2T$ serves. The user uploads $f(P_{L+i}),g(P_{L+i})$ to the server $N_{i}$, $i\in [T]$ and $f(P_{q^{3}-L-T+i}),g(P_{q^{3}-L-T+i})$ to the server $N_{T+i}$, $i\in [L+T]$.

\noindent \textbf{Download Phase:} 
Each server $N_{i}$, $i\in [T]$ computes $-h(P_{L+i})=-f(P_{L+i})g(P_{L+i})$, and each server $N_{T+i}$, $i\in [L+T]$ computes $-h(P_{q^{3}-L-T+i})=-f(P_{q^{3}-L-T+i})g(P_{q^3-L-T+i})$ and sends these values to the user.

\noindent \textbf{User Decoding:} In Lemma~\ref{lem:decodability}, we show that the user can decode $AB=h(P_1)+h(P_{2})+\cdots+h(P_{L})$ from 
$\{-h(P_{L+i})\}_{i=1}^{T}\cup\{-h(P_{q^3-L-T+i})\}_{i=1}^{L+T}$.

\section{Proof of Theorem \ref{theo:scheme}}
We split the proof into lemmas. We show that \textit{HerA} schemes are decodable in Lemma~\ref{lem:decodability} and $T$-secure, in Lemma~\ref{lem:tsecure}. These statements combined prove Theorem~\ref{theo:scheme}. 

\begin{lemma}\label{lem:decodability}
Given a prime power $q$ and positive integers $L$ and $T$ such that $L+T\leq\frac{q^3}{2}$, define $m:=L+T+\frac{q(q-1)}{2}-1$ and let $\mathcal{L}(m P_{\infty})$ be the Riemann-Roch space of the divisor $mP_{\infty}$ on the Hermitian curve $\mathcal{H}_{q}:y^{q}+y=x^{q+1}$. If $A=\left[A_{1}\cdots A_{L}\right]\in\mathbb{F}^{a\times b}_{q^2}$ and $B^{T}=\left[
B^{T}_{1}\cdots B^{T}_{L}\right]\in\mathbb{F}^{c\times b}_{q^2}$. Then, $h(P_1)+\cdots+h(P_{L})$ can be decoded using $L+2T$ servers.
\end{lemma}
\begin{proof}
Let $\left\lbrace P_{i}\right\rbrace_{i=1}^{q^3-L-T}\subseteq\mathcal{H}_{q}(\mathbb{F}_{q^2})\setminus\left\lbrace P_{\infty}\right\rbrace$, $f\in\mathcal{L}(mP_{\infty})$ and $g\in\mathcal{L}(m^{\perp}P_{\infty})$ be polynomials such that
\begin{equation*}
\begin{array}{llr}
f(P_i)=A_{i},&g(P_{i})=B_{i} \forall i\in[L],\\
f(P_{L+i})=R_{i},&g(P_{L+i})=S_{i} \forall i\in[T],\\
&g(P_{L+T+i})=0\forall i\in[ q^3-2(L+T)],
\end{array}
\end{equation*}
using the inner product partitioning $A=\left[
A_{1}\cdots A_{L}\right]$ and $B^{T}=\left[
B^{T}_{1}\cdots B^{T}_{L}\right]$ and uniformly distributed random $\mathbb{F}_{q^2}$-matrices $R_{i},S_{i}$. Therefore, $h(x,y)=f(x,y)g(x,y)$ is such that $h(P_{i})=A_{i}B_{i}$ for all $i\in [L]$. Note that 
$$(f(P_1),\ldots,f(P_{q^3}))\in\mathcal{C}(mP_{\infty})$$
and
$$(g(P_1),\ldots,g(P_{q^3}))\in\mathcal{C}(m^{\perp}P_{\infty})$$
since $f\in\mathcal{L}(mP_{\infty})$ and  $g\in\mathcal{L}(m^{\perp}P_{\infty})$. The dual-code property implies that
\begin{eqnarray*}
0&=&\sum_{i=1}^{q^3}f(P_i)g(P_i)\\
&=&\sum_{i=1}^{L}f(P_i)g(P_i)+\sum_{i=L+1}^{L+T}f(P_i)g(P_i)\\&&+\sum_{i=q^3-L-T+1}^{q^3}f(P_i)g(P_i).
\end{eqnarray*}
So, $\sum_{i=1}^{L}f(P_i)g(P_i)=$
\begin{equation}
-\sum_{i=L+1}^{L+T}f(P_i)g(P_i)-\sum_{i=q^3-L-T+1}^{q^3}f(P_i)g(P_i),
\end{equation}
proving that $AB$ is performed using $L+2T$ servers.
\end{proof}
\begin{definition} The functions $f_1(x,y), \ldots, f_T(x,y)\in\mathcal{L}(mP_\infty)$  and the set $\mathcal{T}\subset \{1, 2, \ldots, q^3\}$ satisfy the \emph{$T$-MDS} condition if \[
\mathbf{F^{(T)}}=
\left(\begin{matrix}
f_{1}(P_{i_1})&f_{1}(P_{i_2})&\cdots& f_{1}(P_{i_T})\\
f_{2}(P_{i_1})&f_{2}(P_{i_2})&\cdots&f_{2}(P_{i_T})\\
\vdots & \vdots & \ddots & \vdots\\
f_{T}(P_{i_1})&f_{T}(P_{i_2})&\cdots& f_{T}(P_{i_T})\\
\end{matrix}\right)
\] has full rank for any different $i_1, i_2, \ldots, i_T\in \mathcal{T}$ and $P_{i_1}, P_{i_2}, \ldots, P_{i_T} \in \mathcal{H}_q(\mathbb{F}_{q^2})\setminus\{P_\infty\}$.
\end{definition}
\begin{lemma}\label{lem:tsecure}
Let  $f(x,y) = \sum_{i=0}^{L+T} f_i(x,y)$ be a polynomial encoding the matrix $A$ satisfying conditions in Lemma \ref{lem:decodability}. If $f_{L+1}(x, y), f_{L+2}(x, y), \ldots, f_{L+T}(x, y)$ and there is $\mathcal{T} \subset \{L+1, \ldots, L+T\}\cup \{q^3-2(L+T), q^3\}$ with $|\mathcal{T}|=L+2T$ satisfying the $T$-MDS condition,
then $I(A; f(P_{i_1}), f(P_{i_2}),\ldots, f(P_{i_T}))=0$. A similar argument holds for $g(x,y)$ and $B$, implying that the HerA scheme is $T$-secure.
\end{lemma}

\begin{proof}
{Since $f(x,y)$ is independent of $B$ and $g(x,y)$ is independent of $A$, proving $T$-security is equivalent to showing that $I(A;f(P_{i_1}), \ldots, f(P_{i_T}))=I(B;g(P_{i_1}), \ldots, g(P_{i_T}))=0$, for any $$\mathcal{T} = \{i_1, \ldots, i_T\}\subset \{L+1, \ldots, L+T\}\cup \{q^3-2(L+T), q^3\}.$$  We prove the claim for $f(x,y)$ since the proof for $g(x,y)$ is analogous.

Since $m = L+T + \frac{q(q-1)}{2} - 1$, $ |I(m)|= L+T$. Therefore, there exists $\left\lbrace P_{1},\ldots,P_{L+T}\right\rbrace\subseteq\mathcal{H}_{q}(\mathbb{F}_{q^2})\setminus\left\lbrace P_{\infty}\right\rbrace$ and  $f(x,y)$ expressed as
\begin{align*}
f(x,y) = \sum_{i=1}^{L+T}f_{i}(x,y)f(P_i)
,
\end{align*}
where each $f_{i}(x,y)\in \mathcal{L}(mP_{\infty})$ with 

\[
f_{i}(x,y) = \left\lbrace\begin{array}{ll}
          1&\text{if } (x,y) = P_i  \\
         0 &\text{if } (x,y) = P_j \text{ and } j\in [L+T]\setminus\{i\}.
    \end{array}\right.
\]

Then,
\begin{align*}
&I(A;f(P_{i_1}), \ldots, f(P_{i_T}))\\
=&H(f(P_{i_1}), \ldots, f(P_{i_T})) - H(f(P_{i_1}), \ldots, f(P_{i_T})|A)\\
\le & \sum_{j \in \mathcal{T}}H(f(P_{j})) - H(f(P_{i_1}), \ldots, f(P_{i_T})|A)\\
=& \frac{Tab}{L}\log(q^2) -\frac{\mathrm{rank}(\mathbf{F^{(T)}})ab}{L}\log(q^2).
\end{align*}

Since the evaluation points $E=\{P_i: i\in \mathcal{T}\}$ are such that $\mathbf{F^{(T)}}$ has full rank for any $T$ different $P_i$'s in $E$, the $f^{(T)}(P_{i_j})$'s are uniformly distributed in the space of the matrices $M_{a\times \frac{b}{L}}(\mathbb{F}_{q^2})$. Thus, 
$H(f^{(T)}(P_{i_1}), \ldots, f^{(T)}(P_{i_T})) = \frac{Tab}{L}\log(q^2)$;
therefore, $I(A;f(P_{i_1}), \ldots, f(P_{i_T}))= 0$.}
\end{proof}

\section{Example: $L=2$, $T=2$ and $q=3$}\label{sec6}
{
\noindent \textbf{Choosing Parameters $L$, $T$ and $m$:} Since $L=T=2$, the \textit{HerA} scheme can be performed on $\mathbb{F}_{3^2}$ ($q=3$), setting $m=6$.

\noindent \textbf{Choosing the Polynomials:} Since $L=2$, the matrices $A\in\mathbb{F}^{a\times b}_{9}$ and $B\in\mathbb{F}^{b\times c}_{9}$ are partitioned as $A=[A_{1} A_{2}]$ and $B^{T}=[B^{T}_{1} B^{T}_{2}]$ such that $AB=A_{1}B_{1}+A_{2}B_{2}$, with $A_{i}\in\mathbb{F}^{a\times\frac{b}{2}}_{9}$, $B_{i}\in\mathbb{F}^{\frac{b}{2}\times c}_{9}$, for $i=1,2$. In order to obtain $2$-security, $R_{1},R_{2}\in\mathbb{F}^{a\times\frac{b}{2}}_{9}$ and $S_{1},S_{2}\in\mathbb{F}^{\frac{b}{2}\times c}_{9}$ are chosen independently and uniformly at random. Choose $\left\lbrace P_{1},P_{2},P_{3},P_{4},P_5,P_6,P_7,P_8\right\rbrace\subseteq\mathcal{H}_{3}(\mathbb{F}_{9})\setminus\left\lbrace P_{\infty}\right\rbrace$, $f\in\mathcal{L}(6P_{\infty})$, $g\in\mathcal{L}(25P_{\infty})$ such that $f(P_{i})=A_{i}$, $g(P_{i})=B_{i}$ for every $i\in [2]$; $f(P_{2+i})=R_{i}$, $g(P_{2+i})=S_{i}$ for every $i\in [2]$ and $g(P_{j})=0$ for all $P_{j}\in\mathcal{H}_{3}(\mathbb{F}_{9})\setminus(\left\lbrace P_{\infty}\right\rbrace\cup\lbrace P_{i}\rbrace_{i=1}^{8})$.

\noindent In this example, for $$P_{1}=(0,0),P_{2}=(0,\delta+1),P_{3}=(1,2), P_{4}=(\delta,1),$$
$$P_5=(2,2), P_6=(\delta+1,2), P_7=(\delta+2,\delta+2),P_8=(2\delta,1)\text{,}$$ with $\delta\in \mathbb{F}_9$ a primitive element, transforming it in an interpolation problem. With the help of the Mathematics Software SageMath\cite{sagemath}, we found a solution for the system of equations leading to the following encoding polynomials
$$f(x,y)= f_{1}(x,y)A_1 + f_{2}(x,y)A_2+f_{3}(x,y)R_1+f_{4}(x,y)R_2$$ where {${f_1(x,y)}=1+\delta x^2+(\delta+1)y$, ${f_{2}(x,y)}=2\delta x+2x^2+(2\delta +2)y$, ${f_3(x,y)}=(\delta+1)x+2\delta x^2$ and ${f_{4}(x,y)}=2x+x^2$} and $$g(x,y)= g_{1}(x,y)B_1 + g_{2}(x,y)B_2+g_{3}(x,y)S_1+g_{4}(x,y)S_2$$ where {${g_{1}(x,y)}=(1+2x+xy+2x^3y+x^4y+x^5y+x^6y+y^2+2x^2y^2+x^3y^2+x^4y^2+x^5y^2)+\delta(2x^4+2x^7+x^8)+(\delta+2)(x^2+x^3)+(2\delta+2)x^6$ and ${g_{2}(x,y)}=(2x^2+2x^3+x^2y+x^3y+x^4y+y^2)+\delta(2x+2x^5+2x^6+2x^7+2xy+2x^5y+xy^2+2x^4y^2+x^5y^2)+(\delta+1)(y+x^2y^2+x^3y^2)+(\delta+2)x^8+(2\delta+1)(x^7y+x^6y)$\\${g_{3}(x,y)}=(x^6y+2x^7y+x^7)+\delta(x^5+x^6+x^3y)+(\delta+1)(x^2+x^3+x^5y^2+x^8)+(\delta+2)(xy+x^2y^2)+2\delta(x^2y+x^5y+xy^2)+(2\delta+2)(x+x^4)$\\ ${g_{4}(x,y)}=(x^2y^2+2x^5y^2)+\delta(x+x^7+x^6y+x^7y)+(\delta+1)(x^3+x^6+x^3y)+(\delta+2)(x^4+x^3y^2)+2\delta(xy+x^5y)+(2\delta+1)(x^2+x^8+xy^2)+(2\delta+2)x^2y$}.
\\

\noindent \textbf{Upload Phase:} The \textit{HerA} scheme uses $6$ servers. The user uploads $f(P_{2+i}),g(P_{2+i})$ to the $i$-th server. 

\noindent \textbf{Download Phase:} 
Each server computes $-h(P_{2+i})=-f(P_{2+i})g(P_{2+i})$ and sends the values back to the user.

\noindent \textbf{User Decoding:} In Lemma~\ref{lem:decodability}, we show that the user can decode $AB=h(P_1)+h(P_{2})$ from 
$$\{-h(P_{2+i})\}_{i=1}^{2}\cup\{-h(P_{4+i})\}_{i=1}^{4}.$$
It can be checked that $f_3(x,y)$, $f_4(x,y)$ and $\mathcal{T}= \{3,4, \ldots, 8\}$ satisfy the $2$-MDS condition. Similarly, $g_3(x,y)$, $g_4(x,y)$ and the same $\mathcal{T}$ satisfy the $2$-MDS condition. Therefore, the scheme is $2$-secure.}

\section{Conclusion}

This paper proposes a secret-sharing-based scheme for secure distributed matrix multiplication. We give a general framework based on Hermitian codes achieving the same recovery threshold as the state-of-the-art in the literature while still using small finite fields, allowing more partitioning and security in a fixed finite field $\mathbb{F}_{q^2}$. 

\bibliographystyle{IEEEtran}
\bibliography{references.bib}

\end{document}